\documentclass[12pt]{article}

\usepackage{amsmath}
\usepackage{amsthm}
\usepackage{amsfonts}
\usepackage{amssymb}
\usepackage{latexsym}
\usepackage{amscd}
\usepackage{stmaryrd}
\usepackage{amsbsy}

\allowdisplaybreaks
\newcommand{\be}{\begin{eqnarray*}}
\newcommand{\en}{\end{eqnarray*}}
\newcommand{\bea}{\begin{eqnarray}}
\newcommand{\ena}{\end{eqnarray}}

\newcommand{\C}{{\mathbb C}}
\newcommand{\Z}{{\mathbb Z}}

\newcommand{\slt}{\mathfrak{sl}_2}

\newcommand{\res}{{\rm res}}

\newcommand{\tr}{{\rm tr}}

\newcommand{\h}{\mathfrak{h}}

\def\sl{\operatorname{\mathfrak{sl}}}

\def\g{\operatorname{\mathfrak{g}}}
\def\tr{\operatorname{tr}}

\numberwithin{equation}{section}

\theoremstyle{plain}
\newtheorem{thm}{Theorem}[section]

\newtheorem{lem}[thm]{Lemma}
\newtheorem{prop}[thm]{Proposition}
\newtheorem{dfn}[thm]{Definition}
\newtheorem{exmp}[thm]{Example}
\newtheorem{re}[thm]{Remark}
\newtheorem{conj}[thm]{Conjecture}

\begin{document}

\title{Confluent KZ equations for $\sl_N$ with Poincar\'e rank $2$ at infinity}
\date{\today}

\author{Hajime Nagoya\footnote{Present address: 
Department of Mathematics, Kobe University,
 Kobe 657-8501, Japan, 
 Research Fellow of the Japan Society for the Promotion of Science} 
 \\ 
Graduate School of Mathematical Sciences, The
University of Tokyo, \\Tokyo 153-8914, Japan
\\
e-mail: nagoya@math.kobe-u.ac.jp
\\
and 
\\
Juanjuan Sun
\\
Graduate School of Mathematical Sciences, The
University of Tokyo, \\Tokyo 153-8914, Japan
\\
email: sunjuan@ms.u-tokyo.ac.jp
}

\maketitle

\begin{abstract}
We construct confluent KZ equations   with Poincar\'e rank 2 at infinity for the case of $\frak{sl}_N$
and the integral representation for the solutions.
Hamiltonians of these confluent KZ equations are derived from suitable quantization of
$d\log \tau$ constructed in the theory of monodromy preserving deformation in \cite{JMU}.
 Our confluent KZ equations can be viewed as a quantization of monodromy preserving deformation
 with Poincar\'e rank 2 at infinity.

\end{abstract}

\section{Introduction}

The KZ equation is a system of linear partial differential
equations with regular singularities,  and it has integral
formulas of hypergeometric type for solutions \cite{SV1}. Further, the KZ
equation  is  a quantization of the Schlesinger equation,  which describes  monodromy preserving deformation
(MPD) of linear differential equations with regular singularities
\cite{H}, \cite{R}, and \cite{Takasaki}.

Irregular singular versions of the KZ equation have been considered in some cases.
Generalized KZ
equations with Poincar\'e rank 1 at infinity for $\slt$ were presented  in \cite{BK}   and later for any simple Lie
algebra in \cite{FMTV}. 
In \cite{JNS}, 
   confluent KZ
equations with an arbitrary Poincar\'e rank were obtained 
for $\slt$. 

A quantization of the Painlev\'e equations with affine Weyl group
actions was proposed in \cite{JNS}, \cite{N1},  and \cite{N2}.
In \cite{JNS},  the  quantum Painlev\'e equations $\mathrm{QP_J}$ 
($\mathrm{J=I,II,III,IV, V}$) 
  were derived formally.
 In all the above mentioned cases, the
solutions to the equations are expressed by integral formulas of
confluent hypergeometric type.

In this paper, we construct confluent KZ equations with Poincar\'e rank 2
at infinity for $\frak{sl}_N$ and  integral formulas of
confluent hypergeometric type for solutions. 
 Further, we discuss a connection between  
 our confluent KZ equations and a 
 quantization of MPD with Poincar\'e rank 2 at infinity. 

Hamiltonians for our confluent KZ equations are
obtained assuming suitable quantization of the function $d
\log \tau$. The function $\tau$ was given in the study of
MPD \cite{JMU}. The function $d \log \tau$ for  regular singularities
gives  the Hamiltonians of the Schlesinger equation.   Moreover,
it was shown in \cite{U} that
Hamiltonians of MPD with Poincar\'e rank 2 at infinity
for $\frak{sl}_2$  are derived from $d \log \tau$.
Therefore, we expect the exact formulas of Hamiltonians of MPD are obtained
 from
$d\log \tau$. Indeed, compatibility condition of the confluent KZ equation
implies that Hamiltonians of MPD (classical case) with Poincar\'e rank 2 at
infinity are derived from $d\log\tau$. 

Hamiltonians for the KZ equations are called the Gaudin Hamiltonians. 
The problem of diagonalization of the Gaudin Hamiltonians is called the Gaudin model \cite{G}.
 It was known  that the eigenvalues and eigenvectors of
 the Gaudin model are derived from the
 integral formulas for solutions to the KZ equation \cite{BF}.
  Irregular singular versions of the Gaudin model were also studied as well as the 
 standard Gaudin model. 
 In  \cite{FFT},
  higher Gaudin Hamiltonians were constructed  using non-highest
 weight representations of affine algebras, and eigenvectors of these Hamiltonians
 were constructed using Wakimoto modules of critical level.
It was explained in \cite{FFT} that there is a connection between the
irregular singular version of the Gaudin model and the geometric Langlands
correspondence.

The remainder of this paper is organized as follows. In section 2, we present
notations of the truncated Lie algebras and define confluent Verma modules. 
In section 3, recalling the
definition of Jimbo-Miwa-Ueno's tau function $\tau$,  we define Hamiltonians
as suitable quantization of $d \log \tau$, and we give confluent KZ equations. 
In section 4, we present an integral formula 
\begin{equation}
\int_{\Gamma} \Phi^{\frac{1}{\kappa}}\omega
\end{equation}
as a solution to the confluent KZ equation. Here $\Phi$ is a scalar multi-valued master function,  
$\Gamma$ is an appropriate cycle, and 
 $\omega$ is a differential form taking values in a tensor product 
 of confluent Verma 
 modules. The master function $\Phi$ and the differential form $\omega$ 
 are straightforward 
generalizations of the respective parts of the integral formulas for the 
general KZ equations in \cite{FMTV} and confluent KZ equations for
$\frak{sl}_2$ in \cite{JNS}. The differential form $\omega$ is
represented in terms of 
 a Poincar\'e-Birkhoff-Witt
basis,  as shown in \cite{M}. 
In section 5, we  discuss a connection between 
our confluent KZ equations and a quantization of monodromy preserving deformation
 with Poincar\'e rank 2 at infinity.

\section{Preliminary}
\subsection{Notation}
Let $\g$ be the complex simple Lie algebra $\sl_{N}$ and let $\h$ be
the Cartan subalgebra. We denote by $\Pi$, $\Delta$, $\Delta_+$, $Q$ and $Q_+$
the set of simple
roots of $\g$,
the root system, the set of positive roots,  the set of the root lattice, and the set of the positive part of
the root lattice, respectively. Explicitly, we have
\begin{align*}
&\Pi=\{\alpha_1,\ldots,\alpha_{N-1}\},\\
&\Delta=\{\pm(\alpha_i+\cdots+\alpha_j)|1\leq i\leq j\leq N-1\},\\
&\Delta_+=\{\alpha_i+\cdots+\alpha_j|1\leq i\leq j\leq N-1\},\\
&Q=\sum_{p=1}^{N-1}\Z\alpha_p,
\quad Q_+=\sum_{p=1}^{N-1}\Z_{\geq0}\alpha_p.
\end{align*}

We shall use the following sets of roots:
\begin{align*}
J_p&=\{\alpha|\alpha\in\Delta_+\cap\left(\alpha_p+Q_+\right)\},\\
R_p&=\{\alpha_p,\alpha_p+\alpha_{p+1},\ldots,\alpha_p+\alpha_{p+1}+\cdots+\alpha_{N-1}\},\\
C_{p}&=\{\alpha_1+\alpha_2+\cdots+\alpha_p,\alpha_2+\cdots+\alpha_p,\ldots,\alpha_p\},
\end{align*}
for $1\leq p\leq N-1$.   We define
the following linear order of all positive roots.
\begin{dfn}
Let $\alpha=\alpha_p+\alpha_{p+1}+\cdots+\alpha_q$ ($p\leq q$), and
$\alpha'=\alpha_{p'}+\alpha_{p'+1}+\cdots+\alpha_{q'}$ ($p'\leq
q'$). We define $\alpha>\alpha'$ if $p<p'$ or $p=p'$ and $q<q'$.
\end{dfn}

The Lie algebra $\g$ has a basis
$\{e_{\alpha},e_{-\alpha}\;(\alpha\in\Delta_+),h_{p}\;(p=1,\ldots,N-1)\}$,
and dual basis
$\{e_{-\alpha},e_{\alpha}\;(\alpha\in\Delta_+),w_p\;(p=1,\ldots,N-1)\}$.
They satisfy the following commutation relations:
\begin{align*}
&[h,h']=0\quad\quad\quad\quad\quad\quad\;\; \text{if}\; h,h'\in\mathfrak{h},\\
\rule[0cm]{.0cm}{0.5cm}&[h,e_{\alpha}]=\alpha(h)e_{\alpha}\quad\quad\quad\quad
\text{if}\;h\in\mathfrak{h},\alpha\in\Delta,\\
\rule[0cm]{.0cm}{0.5cm}&[e_{\alpha},e_{-\alpha}]=h_{\alpha}\quad\quad\quad\quad\quad
\text{if}\;\alpha\in\Delta_+,\\
\rule[0cm]{.0cm}{0.5cm}&[e_{\alpha},e_{\beta}]=0\quad\quad\quad\quad\quad\quad\;
\text{if}\;\alpha,\beta\in\Delta,\alpha+\beta\notin\Delta\cup\{0\},\\
\rule[0cm]{.0cm}{0.5cm}&[e_{\alpha},e_{\beta}]=\epsilon(\alpha,\beta)e_{\alpha+\beta}\quad\quad
\text{if}\;\alpha,\beta,\alpha+\beta\in\Delta.
\end{align*}
Here $h_{\alpha}=h_{a_p}+h_{\alpha_{p+1}}+\cdots+h_{\alpha_q}$ for
$\alpha=\alpha_p+\alpha_{p+1}+\cdots+\alpha_q$ with
$h_{\alpha_p}=h_p$, and the function $\epsilon$ is defined by
\begin{align*}
\epsilon(\alpha,\beta)=
                         \begin{cases}
                           1, & \alpha>\beta, \\
                           -1, & \alpha<\beta.
                         \end{cases}
\end{align*}
 The Casimir operator is defined by
\begin{align*}
\Omega=\sum\limits_{\alpha\in\Delta_+}(e_{\alpha}\otimes
e_{-\alpha}+e_{-\alpha}\otimes
e_{\alpha})+\sum^{N-1}\limits_{p=1}h_p\otimes w_p.
\end{align*}

The homomorphism defined below gives a natural representation of
$\g$ on $\C^N$: for
$\alpha=\alpha_p+\alpha_{p+1}+\cdots+\alpha_{q}$,
$e_{\alpha}=E_{p,q+1}$, $e_{-\alpha}=E_{q+1,p}$;
$h_p=E_{p,p}-E_{p+1,p+1}$ and
$$w_p=\sum^p\limits_{a=1}(1-\frac{p}{N})E_{a,a}-\frac{p}{N}\sum^{N}\limits_{a=p+1}E_{a,a},$$
where $E_{p,q}$ is the matrix with the $(p,q)$-entry $1$ and others
$0$.


\subsection{Module}
In this subsection, we define a confluent Verma module which is
a natural generalization of a Verma module and will be used
in subsequent sections. For $\slt$, a confluent Verma module was
defined in \cite{JNS}.

 Let $V_i$ be Verma
modules of $\g$ with respect to highest weights $\Lambda^{(i)}$ and
highest weight vectors $v_i$, for $i=1,\ldots,n$.

Let $\g_{(2)}$ be the truncated Lie algebra $t\g[t]/(t^3\g[t])$, where
$\g[t]=\g\otimes \mathbb{C}[t]$. We denote $x\otimes t^p$ by $x[p]$. 
Let $V^{(\infty)}$ be defined as a cyclic $\g_{(2)}$-module with highest weights
$\Lambda^{(\infty)}_p$ ($p=1,2$) generated by a vector $v_{\infty}$
such that
\begin{align*}
&e_{\alpha}[1]v_{\infty}=0,
\quad
h_{\alpha}[p]v_{\infty}=(\Lambda^{(\infty)}_p,\alpha)v_{\infty}, 
\quad (\Lambda^{(\infty)}_2,\alpha)\neq 0 
\quad(\forall\alpha\in\Delta_+,\;p=1,2), 
\end{align*}
and $e_{\pm \alpha}[2]$ $(\alpha\in\Delta_+$, $p=1,2$) act as $0$ on $V^{(\infty)}$. 


The action of $\g_{(2)}$ on $V^{(\infty)}$ is simple. The elements $h_{\alpha}[p]$ 
$(\alpha\in\Delta_+$, $p=1,2$) act as scalars on $V^{(\infty)}$. 
The action of $e_{\alpha}[1]$ $(\alpha\in\Delta_+$, $p=1,2$) is only non-commutative 
with the action of $e_{-\alpha}[1]$. 
We denote 
$(\Lambda^{(\infty)}_1,\alpha)$ by 
 $\gamma_\alpha$  and $(\Lambda^{(\infty)}_2,\alpha)$ by 
$\mu_\alpha$.

Let the module $V^{(\infty)}$ be equipped with 
differential operators with respect to $\gamma_k(:=\gamma_{\alpha_k})$ and
$\mu_k(:=\mu_{\alpha_k})$  ($k=1,\ldots,N-1$)  as follows:
Let the differential operators $\partial /\partial \gamma_k$ and 
$\partial /\partial \mu_k$ 
be defined as 
\begin{align*}
&\frac{\partial }{\partial \gamma_k} (e_{\pm \alpha}[1])=0, \quad
\frac{\partial }{\partial \gamma_k} (h_{ \alpha}[1])=1,\quad
\frac{\partial }{\partial \gamma_k} (h_{ \alpha}[2])=0,
\\
&
\frac{\partial }{\partial \mu_k} (e_{\pm \alpha}[1])=\frac{1}
{2\mu_\alpha}(e_{\pm \alpha}[1]) \quad
(\alpha\in J_k), \quad \frac{\partial }{\partial \mu_k} (e_{\pm \alpha}[1])=0 \quad
(\alpha\not\in J_k),
\\
&\frac{\partial }{\partial \mu_k} (h_{ \alpha}[1])=0,\quad
\frac{\partial }{\partial \mu_k} (h_{ \alpha}[2])=1\quad 
(\alpha\in J_k), \quad \frac{\partial }{\partial \mu_k} (h_{ \alpha}[2])=0\quad 
(\alpha\not\in J_k), 
\\
&
\frac{\partial }{\partial \gamma_k}(v_\infty)=
\frac{\partial }{\partial \mu_k}(v_\infty)=0, 
\end{align*}
for $\alpha\in\Delta_+$. Here $x[p]$ ($x=e_{\pm\alpha},h_\alpha$, $p=1,2$)
 is the action on $V^{(\infty)}$. We can verify that the differential operators 
 $\partial /\partial \gamma_k$ and 
$\partial /\partial \mu_k$ preserve the commutation relations of 
the action of $\g_{(2)}$ on $V^{(\infty)}$. Hence, the differential operators 
 $\partial /\partial \gamma_k$ and 
$\partial /\partial \mu_k$ on $V^{(\infty)}$ are well-defined.
We call $V^{(\infty)}$ a confluent Verma module. 
%
%
%
%
%
%
\begin{exmp}
Let $\frak{g}=\frak{sl}_3$. Then a confluent Verma module 
$V^{(\infty)}$ is realized as a polynomial ring $\C[x_1,x_2,x_3]$. 
The actions of $x[p]$ ($x=e_{\pm\alpha},h_\alpha$, $p=1,2$) are expressed as 
\begin{align*}
&e_{\alpha_i}[1]=\mu_i^{\frac{1}{2}} \partial_i, \quad
e_{-\alpha_i}[1]=\mu_i^{\frac{1}{2}} x_i, \quad
e_{\pm\alpha_i}[2]=0
\quad
(i=1,2),
\\
&e_{\alpha_1+\alpha_2}[1]=(\mu_1+\mu_2)^{\frac{1}{2}}\partial_3, \quad
e_{-\alpha_1-\alpha_2}[1]=(\mu_1+\mu_2)^{\frac{1}{2}}x_3,
\quad
e_{\pm(\alpha_1+\alpha_2)}[2]=0,
\\
&h_{\alpha_i}[1]=\gamma_i,\quad h_{\alpha_i}[2]=\mu_i,\quad (i=1,2),
\end{align*}
where $\partial_i=\partial/ \partial x_i$.
\end{exmp}

We consider $(\oplus_{i=1}^n\g\oplus\g_{(2)})$-module
\begin{equation}
V=V_1\otimes\cdots V_n\otimes  V^{(\infty)}
\end{equation}
with $\pmb{v}=v_1\otimes \cdots \otimes v_n\otimes v_\infty$. 
For each $i=1,\ldots,n$,  denote by $x^{(i)}:V\to V $ ($x\in\g$) the linear operator 
acting as $x$   
 on $i$th tensor factor
 $V_i$ and as identities on the others. For $x[p]\in\g_{(2)}$, 
  denote by $x^{(\infty)}[p]:V\to V$ the linear operator acting as 
  $x[p]$ on $V^{(\infty)}$ 
 and as identities on the others.

Let 
$\Lambda=\sum^n\limits_{i=1}\Lambda^{(i)}$, ${\bf m}=(m_1,\ldots,m_{N-1})\in(\mathbb{Z}_{\ge 0})^{N-1}$, and a map
$\alpha: (\mathbb{Z}_{\ge 0})^{N-1}\longrightarrow Q_+$ be defined as 
$\alpha({\bf m})=\sum_{p=1}^{N-1} m_p\alpha_p$. We denote by $V_{\bf{m}}$ the
weight space of $V$ with weight $\Lambda-\alpha(\bf{m})$
corresponding to $\sum^n\limits_{i=1}h^{(i)}+h^{(\infty)}[0]$, ($h\in\frak{h}$), 
that is, 
\begin{equation*}
V_{\bf{m}}=\left\{
x\in V \Big| \left(\sum^n_{i=1}h^{(i)}+h^{(\infty)}[0]\right) (x)=
\left(\Lambda-\alpha(\bf{m})\right)(h)x, \ h\in\frak{h}
\right\}.
\end{equation*}
Here, for $h\in\frak{h}$, $h^{(\infty)}[0]$ is defined as 
\begin{equation*}
h^{(\infty)}[0]=-\sum_{\alpha\in\Delta_+}
\frac{\alpha(h)}{\mu_\alpha}
e^{(\infty)}_{-\alpha}[1]
e^{(\infty)}_{\alpha}[1].
\end{equation*} 


\section{Confluent KZ equation}
\subsection{Hamiltonian}
Let $$\frac{\partial Y}{\partial x}=A(x)Y$$ be a system of linear
ordinary differential equations, where $A(x)$ is a rational matrix.
In \cite{JMU} and \cite{JM}, Jimbo, Miwa and Ueno developed a
general theory of monodromy preserving deformation. They derived
non-linear deformation equations and proved their complete
integrability. They also gave explicit formula for a $1$-form
$\omega$ expressed in terms of the coefficients of $A(x)$, with the
property $d\omega=0$ for each solution of the deformation equations. In
\cite{Boalch}, Boalch showed that the deformation equations are
Hamiltonian systems. However, the explicit formula for the
Hamiltonian has not been given. Here we show that after an
appropriate ordering the $1$-form $\omega$ indeed provides the
Hamiltonians in the quantum case at least up to Poincar\'e rank $2$. In the
follwoing we give the concrete construction of the Hamiltonians. For
that purpose, we recall the procedure of \cite{JMU}.

 Let us
consider the following system of linear ordinary differential
equations for an $N\times N$ matrix $Y(z)$ on $\mathbb{P}^{1}$,
\begin{equation}
\frac{d Y}{d z}=A(z)Y \label{A1}
\end{equation}
where
$$A(z)=\sum^{n}\limits_{i=1}\sum^{r_i}\limits_{p=0}\frac{A^{(i)}_{p}}{(z-z_i)^{p+1}}
-\sum^{r}\limits_{p=1}z^{p-1}B_p$$ with $A^{(i)}_{p}$
($i=1,\ldots,n$) and $B_p$ ($p=1,\ldots,r$) are $N\times N$
matrices, and $r_i,r$ are non-negative integers.

 We assume
\begin{align*}
B_r=\begin{pmatrix} t_1&&\\
&\ddots&\\
&&t_N
\end{pmatrix}
\end{align*}
with $t_i\neq t_j$ for $i\neq j$. For simplicity, we consider the
case where all $r_i=0$ ($i=1,\ldots,n$) and $r$ is a positive
integer.

Now we consider the local solution of \eqref{A1} at $z=\infty$.
\begin{prop}[\cite{JMU}]
There exists a unique formal series $Y(z)$ at $z=\infty$ of the
following form:
\begin{equation}
Y(z)=\hat{Y}(z)e^{T(z)}
\end{equation}
which solves \eqref{A1} i.e.
\begin{equation}
\frac{d}{dz}\hat{Y}(z)=A(z)\hat{Y}(z)-\hat{Y}(z)\frac{d}{dz}T(z).\label{A3}
\end{equation}
Here $T(z)$ is a diagonal matrix of the form
\begin{equation}
T(z)=-\sum^{r}\limits_{p=1}T_pz^p/p-T_0\log z,
\end{equation} with $T_r=B_r$ and $\hat{Y}(z)$ is a formal power series at
$z=\infty$:
\begin{equation}
\hat{Y}(z)=\sum^{\infty}\limits_{p=0}Y_pz^{-p},
\end{equation}with $Y_0=1$.
\end{prop}

The formal power series
$\hat{Y}(z)=1+\sum^{\infty}\limits_{p=1}Y_pz^{-p}$ is uniquely
factorized into
\begin{equation}
\hat{Y}(z)=F(z)D(z),\label{A6}
\end{equation}where
$$F(z)=1+\sum^{\infty}\limits_{p=1}F_pz^{-p},\quad
 D(z)=1+\sum^{\infty}\limits_{p=1}D_pz^{-p}.$$ We have $F(z)-1$ is diagonal
 free, and $D(z)$ is diagonal. Substituting \eqref{A6} to \eqref{A3} we obtain
 \begin{equation}
 \frac{d}{dz}F(z)+F(z)\frac{d}{dz}(\log
 D(z)+T(z))=A(z)F(z).\label{A7}
 \end{equation}
 Taking the diagonal part, we obtain
 \begin{equation}
 \frac{d}{dz}(\log D(z)+T(z))=(A(z)F(z))_D
 \end{equation}
 Then \eqref{A7} reads as
 \begin{equation}
 \frac{d}{dz}F(z)+F(z)(A(z)F(z))_D=A(z)F(z).
 \end{equation}
Here and below, we denote by
$X_D=(\delta_{i,j}X_{jj})_{i,j=1,\ldots,N}$ (resp. $X_{OD}=X-X_D$)
the diagonal part (resp. off-diagonal part) of a matrix
$X=(X_{i,j})_{i,j=1,\ldots,N}$.

The authors of \cite{JMU} derived a complete integrable non-linear
deformation equations whose deformation parameters are $t^{(k)}_{\nu}$,
where
\begin{align*}
T_k=\begin{pmatrix} t^{(k)}_1&&\\
&\ddots&\\
&&t^{(k)}_N
\end{pmatrix}\quad(k=1,\ldots,r).
\end{align*}
 They also gave an explicit $1$-form $\omega$ with the property $d\omega=0$
 for each solution of the deformation equations. The formula for $\omega$ reads as
\begin{align}
&\omega=-\res_{z=\infty}\tr\hat{Y}^{-1}(z)\frac{\partial\hat{Y}(z)}{\partial
z}d'T(z),
\end{align}
here $d'$ is the exterior differentiation with respect to the
parameters $t$.

For $k=1,\ldots,r$, we set
\begin{equation*}
\omega^{(k)}=\sum^{N}\limits_{p=1}H^{(k)}_{p}\text{d}t^{(k)}_{p}
\end{equation*}
and
\begin{equation*}
\omega=\sum^{r}\limits_{k=1}\sum^{N}\limits_{p=1}H^{(k)}_{p}\text{d}t^{(k)}_{p}
=\sum^r\limits_{k=1}\omega^{(k)}.
\end{equation*}


Set $T_p=0$ for $p<0$, and
$B_{-p}=\sum^n\limits_{i=1}A_0^{(i)}z_i^{p}$ for $p\geq0$. Rewriting
\eqref{A7} in terms of $F_p$, $D_p$, $B_p$ and $T_p$, we get the
form of $\omega$ which are presented by the matrix elements of
$B_{p}s$. We give the explicit form for the case $r=2$. We have
\begin{align}
\quad\omega^{(1)}&=\sum^{N}\limits_{p=1}\left(\sum_{k\neq p,s\neq
p}\frac{(B_1)_{pk}(B_1)_{ks}(B_1)_{sp}}{(t^{(2)}_p-t^{(2)}_k)(t^{(2)}_p-t^{(2)}_s)}-(B_{-1})_{pp}\right.
\label{formu1}
\\
&\left.\quad-\sum\limits_{k\neq
p}\frac{(B_1)_{pk}(B_{0})_{kp}+(B_0)_{pk}
(B_1)_{kp}}{t^{(2)}_p-t^{(2)}_k}+\sum\limits_{k\neq
p}\frac{(B_1)_{pk}(B_1)_{kp}}{(t^{(2)}_p-t^{(2)}_k)^2}\right)dt^{(1)}_p,
\nonumber
\end{align}
and
\begin{align}
\omega^{(2)}
=&\frac{1}{2}\sum^{N}\limits_{p=1}\left\{-\sum\limits_{k\neq p,s\neq
p}\frac{(B_1)_{pk}(B_0)_{ks}(B_1)_{sp}}{(t^{(2)}_p-t^{(2)}_s)(t^{(2)}_p-t^{(2)}_k)}-\sum\limits_{k,s\neq
p}\frac{(B_0)_{pk}(B_1)_{ks}(B_1)_{sp}}{(t^{(2)}_p-t^{(2)}_k)(t^{(2)}_p-t^{(2)}_s)}\right.
\label{formu2}
\\
&-\sum\limits_{k\neq p,s\neq
p}\frac{(B_1)_{pk}(B_1)_{kp}(B_1)_{ps}(B_1)_{sp}}{(t^{(2)}_p-t^{(2)}_s)(t^{(2)}_p-t^{(2)}_k)^2}
+\sum\limits_{k\neq
p}\frac{(B_1)_{pk}(B_1)_{kp}}{(t^{(2)}_p-t^{(2)}_k)^2}(B_0)_{pp}\nonumber\\
&+\sum\limits_{k,s,r\neq
p}\frac{(B_1)_{pk}(B_1)_{ks}(B_1)_{sr}(B_1)_{rp}}{(t^{(2)}_p-t^{(2)}_k)(t^{(2)}_p-t^{(2)}_s)(t^{(2)}_p-t^{(2)}_r)}-
\sum\limits_{k,s\neq
p}\frac{(B_1)_{pk}(B_1)_{ks}(B_0)_{sp}}{(t^{(2)}_p-t^{(2)}_k)(t^{(2)}_p-t^{(2)}_s)}
\nonumber
\\
&-\sum\limits_{k,s\neq
p}\frac{(B_1)_{pk}(B_1)_{ks}(B_1)_{sp}}{(t^{(2)}_p-t^{(2)}_k)(t^{(2)}_p-t^{(2)}_s)^2}t^{(1)}_p
-\sum\limits_{k,s\neq
p}\frac{(B_1)_{pk}(B_1)_{ks}(B_1)_{sp}}{(t^{(2)}_p-t^{(2)}_k)^2(t^{(2)}_p-t^{(2)}_s)}t^{(1)}_p\nonumber\\
&+\sum\limits_{k\neq
p}\frac{(B_1)_{pk}(B_0)_{kp}+(B_0)_{pk}(B_1)_{kp}}{(t^{(2)}_p-t^{(2)}_k)^2}t^{(1)}_p
+\sum\limits_{k\neq
p}\frac{(B_1)_{pk}(B_1)_{kp}}{(t^{(2)}_p-t^{(2)}_k)^3}(t^{(1)}_p)^2\nonumber\\
&\left.+\sum\limits_{k\neq
p}\frac{(B_0)_{pk}(B_0)_{kp}}{t^{(2)}_p-t^{(2)}_k}-\sum\limits_{k\neq
p}\frac{((B_1)_{pk}(B_{-1})_{kp}+(B_{-1})_{pk}(B_1)_{kp})}{t^{(2)}_p-t^{(2)}_k}
-(B_{-2})_{pp}\right\}dt^{(2)}_p. \nonumber
\end{align}

Here, $(B_i)_{pq}$ are the $(p,q)$-entries of the matrix $B_i$.

\vskip5pt

In the following, we consider the quantum case with restricting to
$\sl_N$ for $r=2$.

For $p=1,\ldots,N$, we define $\bar{H}^{(1)}_{p}$
(resp.$\bar{H}^{(2)}_{p}$) in terms of $H^{(1)}_{p}$ (resp.
$H^{(2)}_{p}$) as follows. First, we put $({B}_i)_{pq}$ with $p\leq
q$ to the right and the others to the left. Second,for
$\alpha=\alpha_p+\cdots+\alpha_{q}$, we substitute
$e^{(i)}_{\alpha}$, $e^{(i)}_{-\alpha}$, $h^{(i)}_p$,
 $e^{(\infty)}_{\alpha}[1]$, $e^{(\infty)}_{-\alpha}[1]$ and
 $h_p^{(\infty)}[1]$, $h_p^{(\infty)}[2]$ for the matrix realizations $(A^{(i)}_0)_{pq+1}$, $(A^{(i)}_0)_{q+1p}$, $(A^{(i)}_0)_{pp}-(A^{(i)}_0)_{p+1p+1}$,
 $(B_1)_{pq+1}$, $(B_1)_{q+1p}$  and $(B_1)_{pp}-(B_1)_{p+1p+1}$, $(B_2)_{pp}-(B_2)_{p+1p+1}$, respectively. The definitions induce that
 $\bar{H}^{(1)}_{p}$, $\bar{H}^{(2)}_{p}$ are actions of elements in the
 algebra $\oplus_{i=1}^n \g \oplus \g_{(2)}$ on the modules $V_1\otimes\cdots V_n\otimes V^{(\infty)}$

 For the actions of $h_p^{(\infty)}[1]=\gamma_p$, $h_p^{(\infty)}[2]=\mu_p$ on the $\g_{(2)}$-module $V^{(\infty)}$, we
 define the following Hamiltonians
\begin{dfn}
For $1\leq p\leq N-1$, we define
 \begin{align*}
 &\mathcal{H}^{(1)}_p=\sum^p_{j=1}(1-\frac{p}{N})\bar{H}^{(1)}_j
 -\frac{p}{N}\sum^N\limits_{j=p+1}\bar{H}^{(1)}_j,\\
&\mathcal{H}^{(2)}_p=\sum^p\limits_{j=1}(1-\frac{p}{N})\bar{H}^{(2)}_j
 -\frac{p}{N}\sum^N\limits_{j=p+1}\bar{H}^{(2)}_j
 +\sum_{\alpha,\beta\in J_p, \atop \alpha-\beta\in\Delta_+}
 \frac{\mu_{\alpha-\beta}}{\mu_\alpha^2\mu_\beta}e^{(\infty)}_{-\alpha}[1]e^{(\infty)}_\alpha[1],
 \end{align*}
where the last term in the second line is called the complement
term.
 More precisely, we have
 \begin{align*}
\mathcal{H}^{(1)}_p=&-\sum^n\limits_{i=1}z_iw^{(i)}_p-\sum\limits_{\alpha\in
J_p}\frac{1}{\mu_\alpha}
\left(e^{(\infty)}_{-\alpha}[1]E_{\alpha}+e^{(\infty)}_{\alpha}[1]E_{-\alpha}\right)
-\sum\limits_{\alpha\in J_p}\frac{\gamma_\alpha}{\mu_\alpha^2}
e^{(\infty)}_{-\alpha}[1]e^{(\infty)}_{\alpha}[1]
\\
&+\sum\limits_{\alpha,\alpha+\beta\in J_p,\atop
\beta\in\Delta_+}\frac{\epsilon(\alpha,\beta)}
{\mu_\alpha\mu_{\alpha+\beta}}
\left(e^{(\infty)}_{\alpha}[1]e^{(\infty)}_{\beta}[1]
e^{(\infty)}_{-\alpha-\beta}[1]
+e^{(\infty)}_{-\alpha}[1]e^{(\infty)}_{-\beta}[1]e^{(\infty)}_{\alpha+\beta}[1]\right),
\\
2\mathcal{H}^{(2)}_p=&-\sum^n\limits_{i=1}z_i^2w^{(i)}_p-\sum\limits_{\alpha\in
J_p}\frac{1}{\mu_\alpha}
\left(e^{(\infty)}_{-\alpha}[1]E_{-\alpha}(-1)+e^{(\infty)}_{\alpha}[1]E_{-\alpha}(-1)\right)
+\sum\limits_{\alpha\in
J_p}\frac{1}{\mu_\alpha}E_{-\alpha}E_{\alpha}
\\
&+\sum\limits_{\alpha\in
J_p}\frac{1}{\mu_\alpha^2}e^{(\infty)}_{-\alpha}[1]e^{(\infty)}_{\alpha}[1]H_{\alpha}
+\sum\limits_{\alpha\in J_p}\frac{\gamma_\alpha}{\mu_\alpha^2}
(e^{(\infty)}_{-\alpha}[1]E_{\alpha}+e^{(\infty)}_{\alpha}[1]E_{-\alpha})
\\
&-\sum\limits_{\alpha,\beta\in J_p,\atop
\alpha-\beta\in\Delta}\frac{\epsilon(\alpha,|\alpha-\beta|)}{\mu_\alpha
\mu_{\beta}} \left(E_{\alpha}e^{(\infty)}_{\beta-\alpha}[1]
e^{(\infty)}_{-\beta}[1] +e^{(\infty)}_{\alpha}[1]E_{\beta-\alpha}
e^{(\infty)}_{-\beta}[1]
+e^{(\infty)}_{\alpha}[1]e^{(\infty)}_{\beta-\alpha}[1]E_{-\beta}
\right)
\\
&+\sum\limits_{\alpha\in J_p}\frac{\gamma_\alpha^2}{\mu_\alpha^3}
e^{(\infty)}_{-\alpha}[1]e^{(\infty)}_{\alpha}[1]-\sum\limits_{\alpha,\beta\in
J_p,\atop\alpha-\beta\in\Delta}\frac{\epsilon(\alpha,|\alpha-\beta|)}{\mu_{\alpha}\mu_{\beta}}
\left(\frac{\gamma_{\alpha}}{\mu_{\alpha}}+\frac{\gamma_{\beta}}{\mu_{\beta}}\right)e^{(\infty)}_{\alpha}[1]e^{(\infty)}_{\beta-\alpha}[1]e^{(\infty)}_{-\beta}[1]
\\
&+\sum\limits_{\alpha,\beta,\gamma\in
J_p,\atop\alpha-\beta,\beta-\gamma\in\Delta}\frac{g(\alpha,\gamma)}{\mu_{\alpha}\mu_{\beta}\mu_{\gamma}}
e^{(\infty)}_{-\alpha}[1]e^{(\infty)}_{\alpha-\beta}[1]e^{(\infty)}_{\beta-\gamma}[1]e^{(\infty)}_{\gamma}[1]
\\
&-\sum_{\alpha,\beta,\gamma\in
J_p,\beta-\alpha-\gamma\in\Delta,\atop\alpha-\beta,\beta-\gamma\in\Delta\cup\{0\},\alpha\geq\gamma}
\frac{1}{\mu_{\alpha}\mu_{\beta}\mu_{\gamma}}
e^{(\infty)}_{-\beta}[1]e^{(\infty)}_{\beta-\alpha-\gamma}[1]e^{(\infty)}_{\alpha}[1]e^{(\infty)}_{\gamma}[1],
\end{align*}
where $X_{\alpha}=\sum_{i=1}^n x^{(i)}_{\alpha}$,
$X_{\alpha}(-1)=\sum_{i=1}^n x^{(i)}_{\alpha}z_i$
($\alpha\in\Delta_+$), and
\begin{align*}
|\alpha-\beta|&=\left\{
\begin{array}{ll}
\alpha-\beta\;\quad\text{if}\;\alpha-\beta\in Q_+\\
\beta-\alpha\;\quad\text{if}\;\beta-\alpha\in Q_+,
\end{array}\right.
\\
g(\alpha,\gamma)&=\left\{
\begin{array}{ll}
1\quad\quad\quad\text{if}\;\alpha-\gamma\in\Delta\cup\{0\},
\\
-1\quad\quad\text{others}.
\end{array}\right.
\end{align*}
\end{dfn}

\begin{exmp}\label{ex hami}
Let $\g=\sl_3$ and $n=0$. 
Then, Hamiltonians $\mathcal{H}^{(1)}_1$ and $\mathcal{H}^{(2)}_1$ are 
\begin{align*}
\mathcal{H}^{(1)}_1=&-\frac{\gamma_1}{\mu^2_1}e^{(\infty)}_{-\alpha_1}[1]e^{(\infty)}_{\alpha_1}[1]
-\frac{\gamma_1+\gamma_2}{(\mu_1+\mu_2)^2}e^{(\infty)}_{-\alpha_1-\alpha_2}[1]e^{(\infty)}_{\alpha_1+\alpha_2}[1]
\\
&+\frac{1}{\mu_1(\mu_1+\mu_2)}\left(e^{(\infty)}_{-\alpha_1}[1]e^{(\infty)}_{-\alpha_2}[1]e^{(\infty)}_{\alpha_1+\alpha_2}[1]
+e^{(\infty)}_{-\alpha_1-\alpha_2}[1]e^{(\infty)}_{\alpha_1}[1]e^{(\infty)}_{\alpha_2}[1]\right),
\end{align*}
and
\begin{align*}
2\mathcal{H}^{(2)}_1&=-\frac{1}{\mu^3_1}(e^{(\infty)}_{-\alpha_1}[1])^2
(e^{(\infty)}_{\alpha_1}[1])^2
-\frac{1}{(\mu_1+\mu_2)^3}(e^{(\infty)}_{-\alpha_1-\alpha_2}[1])^2
(e^{(\infty)}_{\alpha_1+\alpha_2}[1])^2
\\
&
+\frac{\gamma^2_1}{\mu^3_1}e^{(\infty)}_{-\alpha_1}[1]e^{(\infty)}_{\alpha_1}[1]
+\frac{(\gamma_1+\gamma_2)^2}{(\mu_1+\mu_2)^3}e^{(\infty)}_{-\alpha_1-\alpha_2}[1]e^{(\infty)}_{\alpha_1+\alpha_2}[1]
\\
&-\frac{1}{\mu_1(\mu_1+\mu_2)}\left(\frac{\gamma_1}{\mu_1}+\frac{\gamma_1+\gamma_2}{\mu_1+\mu_2}
\right)\left(e^{(\infty)}_{-\alpha_1}[1]e^{(\infty)}_{-\alpha_2}[1]e^{(\infty)}_{\alpha_1+\alpha_2}[1]
+e^{(\infty)}_{-\alpha_1-\alpha_2}[1]e^{(\infty)}_{\alpha_1}[1]e^{(\infty)}_{\alpha_2}[1]\right)\nonumber
\\
&+\frac{1}{\mu^2_1(\mu_1+\mu_2)}\left(e^{(\infty)}_{-\alpha_1}[1]e^{(\infty)}_{\alpha_1}[1]e^{(\infty)}_{-\alpha_2}[1]e^{(\infty)}_{\alpha_2}[1]
-e^{(\infty)}_{-\alpha_1}[1]e^{(\infty)}_{\alpha_1}[1]e^{(\infty)}_{-\alpha_1-\alpha_2}[1]e^{(\infty)}_{\alpha_1+\alpha_2}[1]\right)
\\
&+\frac{1}{\mu_1(\mu_1+\mu_2)^2}\left(e^{(\infty)}_{\alpha_2}[1]e^{(\infty)}_{-\alpha_2}[1]e^{(\infty)}_{-\alpha_1-\alpha_2}[1]e^{(\infty)}_{\alpha_1+\alpha_2}[1]
-e^{(\infty)}_{-\alpha_1-\alpha_2}[1]e^{(\infty)}_{\alpha_1+\alpha_2}[1]e^{(\infty)}_{-\alpha_1}[1]e^{(\infty)}_{\alpha_1}[1]\right). \nonumber
\end{align*}
\end{exmp}

\subsection{Confluent KZ equation}

Now we give the following differential equations, which we call a
confluent KZ equation:
\begin{align}
&\kappa\frac{\partial u}{\partial
z_i}=G^{(i)}_{-1}u\quad(i=1,\ldots,n),\label{CKZ1}\\
&\kappa\frac{\partial
u}{\partial\gamma_p}=\mathcal{H}^{(1)}_pu\quad(p=1,\ldots,N-1),\label{CKZ2}\\
&\kappa\frac{\partial
u}{\partial\mu_p}=\mathcal{H}^{(2)}_pu\quad(p=1,\ldots,N-1), \label{CKZ3}
\end{align}
where $\kappa\in\mathbb{C}$ and the unknown function
$u(z_1,\ldots,z_n,\gamma_1,\ldots,\gamma_{N-1},\mu_1,\ldots,\mu_{N-1})$ takes value in $V$ and
 the Gaudin Hamiltonians $G^{(i)}_{-1}$ are given by
\begin{align*}
G^{(i)}_{-1}=\sum\limits_{1\leq j\leq n,\atop j\neq
i}\frac{\Omega^{(i,j)}}{z_i-z_j}-\sum^{N-1}\limits_{p=1}\gamma_{p}w^{(i)}_{p}
-\sum_{\alpha\in\Delta} e^{(\infty)}_{\alpha}[1]e^{(i)}_{-\alpha}
-z_i\sum^{N-1}\limits_{p=1}\mu_{p}w^{(i)}_{p}\quad(i=1,\ldots,n).
\end{align*}

\begin{conj}
The confluent KZ equation \eqref{CKZ1}-\eqref{CKZ3} satisfies the compatibility condition, 
that is, we have 
\begin{align}
&\left[
\kappa\frac{\partial }{\partial z_i}-G^{(i)}_{-1}, 
\kappa\frac{\partial }{\partial z_j}-G^{(j)}_{-1}
\right]=0\quad (i,j=1,\ldots,n),
\\
&
\left[\kappa\frac{\partial }{\partial z_i}-G^{(i)}_{-1}, 
\kappa\frac{\partial}{\partial\gamma_p}-\mathcal{H}^{(1)}_p
\right]=0\quad (i=1,\ldots,n,\ p=1,\ldots, N-1),
\\
&\left[\kappa\frac{\partial }{\partial z_i}-G^{(i)}_{-1}, 
\kappa\frac{\partial}{\partial\mu_p}-\mathcal{H}^{(2)}_p
\right]=0\quad (i=1,\ldots,n,\ p=1,\ldots, N-1),
\\
&
\left[
\kappa\frac{\partial}{\partial\gamma_p}-\mathcal{H}^{(1)}_p, 
\kappa\frac{\partial}{\partial\gamma_q}-\mathcal{H}^{(1)}_q
\right]=0\quad (p,q=1,\ldots, N-1),
\\
&
\left[
\kappa\frac{\partial}{\partial\gamma_p}-\mathcal{H}^{(1)}_p,
\kappa\frac{\partial}{\partial\mu_q}-\mathcal{H}^{(2)}_q
\right]=0\quad (p,q=1,\ldots, N-1),
\\
&\left[
\kappa\frac{\partial}{\partial\mu_p}-\mathcal{H}^{(2)}_p,
\kappa\frac{\partial}{\partial\mu_q}-\mathcal{H}^{(2)}_q
\right]=0\quad (p,q=1,\ldots, N-1). 
\end{align}
\end{conj}

\begin{re}
Using a software SINGULAR \cite{SINGULAR} (see Appendix), we can confirm that the above conjecture is true when $N$ is less than $7$. 
\end{re}
\section{Integral Formula}
We present integrable formulas taking values in $V_{\bf{m}}$.
Let $S_p$ be an index set $\{1,\ldots,m_p\}$ for each $p$ ($1\le p\le N-1$). We
prepare
the following integration variables
\begin{equation}
\left\{t_a^{(p)} | 1\le p\le N-1, \ a\in S_p\right\}.
\end{equation}
We define the master function of an integrable formula as follows.
\begin{align*}
\Phi(z,t)=&\prod\limits_{1\leq i<j\leq
n}(z_i-z_j)^{(\Lambda^{(i)},\Lambda^{(j)})}\prod\limits_{1\leq j\leq
n\atop 1\leq p\leq
N-1}\prod^{m_p}\limits_{a=1}(t^{(p)}_{a}-z_j)^{-(\alpha_p,\Lambda^{(j)})}
\nonumber
\\
&\times\prod\limits_{1\leq p<q\leq N-1}\prod\limits_{1\leq a\leq
m_p,\atop 1\leq b\leq
m_q}(t^{(p)}_a-t^{(q)}_b)^{(\alpha_p,\alpha_q)}\prod^{N-1}\limits_{p=1}
\prod\limits_{1\leq a<b\leq
m_p}(t^{(p)}_a-t^{(p)}_b)^{(\alpha_p,\alpha_p)}\nonumber
\\
&\times
\exp\left(-\sum^n\limits_{i=1}\left((\Lambda^{(\infty)}_1,\Lambda^{(i)})z_i+
(\Lambda^{(\infty)}_2,\Lambda^{(i)})\frac{z_i^2}{2}\right)\right)\nonumber
\\
&\times\exp\left(\sum^{N-1}\limits_{p=1}\sum^{m_p}\limits_{a=1}\left((\Lambda^{(\infty)}_1,\alpha_p)t^{(p)}_a
+(\Lambda^{(\infty)}_2,\alpha_p)\frac{\left(t^{(p)}_a\right)^2}{2}\right)\right).
\end{align*}

Next, we establish the $\omega$ part consisting of vectors $e_{-\alpha}^{(i)}$, $e_{-\alpha}^{(\infty)}[1]$
($i=1,\ldots,n$, $\alpha\in\Delta_+$).
For each $\alpha=\alpha_i+\cdots+\alpha_j\in\Delta_+$ and
$a=(a_i,\ldots,a_j)$ ($a_k\in S_k$, $i\le k\le j$),
we set
\begin{equation}
e_{-\alpha}(t^{\alpha}_a)=\frac{1}{(t^{(j)}_{a_j}-t^{(j-1)}_{a_{j-1}})\cdots(t^{(i+1)}_{a_{i+1}}-t^{(i)}_{a_i})}
\left(\sum^n_{s=1}
\frac{e_{-\alpha}^{(s)}}{t^{(i)}_{a_i}-z_s}-e^{(\infty)}_{-\alpha}[1]\right).
\end{equation}

Let $K=(k_{\alpha_1},k_{\alpha_1+\alpha_2},\ldots,k_{\alpha_{N-1}})
\in(\mathbb{Z}_{\ge 0})^{\frac{1}{2}N(N-1)}$ and
\begin{align}
S({\bf m})=\left\{K|\sum\limits_{\alpha\in J_p}k_{\alpha}=m_p\;
(p=1,\ldots,N-1)\right\}.
\end{align}
For $K\in S({\bf m})$, let
\begin{align*}
A(K)=(a(\alpha_1,1),\ldots,a(\alpha_1,k_{\alpha_1}),a(\alpha_1+\alpha_2,
1),
\ldots,a(\alpha_{N-1},k_{\alpha_{N-1}})),
\end{align*}
where
$a(\alpha,l)=(a_i(\alpha,l),\ldots,a_j(\alpha,l))$ ($\alpha=\alpha_i+\cdots+\alpha_j$,
$a_k(\alpha,l)\in S_k$, $i\le k\le j$) such that $a_p(\alpha,l)\neq a_p(\beta,l')$ for any
$(\alpha,l)\neq(\beta,l')$
($\alpha,\beta \in J_p$, $1\le l\le k_\alpha$ and $1\le l'\le k_\beta$)
and $\{a_p(\alpha,l)|\alpha\in J_p, 1\le l \le k_\alpha\}=S_p $.

We assign $a(\alpha,l)$ to $e_{-\alpha}
(t^{\alpha}_{a(\alpha,l)})$ and we set
\begin{equation}\label{form}
f_K=\sum_{A(K) }\prod e_{-\alpha}
(t^{\alpha}_{a(\alpha,l)}),
\end{equation}
where the summation is over all  $A(K)$  and
the order for the factor in each product is the same as  the
linear order $\alpha_1>\alpha_1+\alpha_2>\cdots>\alpha_{N-1}$.

We define $\omega_{{\bf{m}}}$ as
\begin{equation}
\omega_{{\bf{m}}}=\sum\limits_{K\in S(\bf{m})}f_K \pmb{v}.
\end{equation}



\begin{re}
If $N=2$, then  $\Phi(z,t)$ and $\omega_{\bf{m}}$ are exactly those
defined in \cite{JNS}. If $e_\alpha^{(\infty)}[1]$ ($\alpha\in\Delta$)
and $\Lambda_2^{(\infty)}$ are zero,   then $\omega_{\bf{m}}$ and
$\Phi(z,t)$ are equivalent to those defined in \cite{FMTV}.
\end{re}
\begin{exmp}\label{ex int}
We give an example for the case of $\sl_3$. Let ${\bf m}=(1,1)$ and
$n=0$, then, 
\begin{equation}
\Phi=\left(t^{(1)}-t^{(2)}\right)^{-1} \exp \left(
\gamma_1t^{(1)}+\mu_1\frac{\left(t^{(1)}\right)^2}{2} \right) \exp
\left( \gamma_2t^{(2)}+\mu_2\frac{\left(t^{(2)}\right)^2}{2}
\right),
\end{equation}
 and
 \begin{equation}
\omega_{\bf
m}=\left(e^{(\infty)}_{-\alpha_1}[1]e^{(\infty)}_{-\alpha_2}[1]
-\frac{e^{(\infty)}_{-\alpha_1-\alpha_2}[1]}{t^{(2)}-t^{(1)}}\right)\pmb{v}.
\end{equation}

\end{exmp}

\begin{thm}\label{thm solution}
With an appropriate choice of cycles $\Gamma$, the function
\begin{equation}\label{eq solution}
u=\int_{\Gamma}\prod\limits_{1\leq
p\leq N-1,\atop 1\leq a\leq
m_p}dt^{(p)}_a\Phi^{1/\kappa}(z,t)\omega_{\bf{m}}
\end{equation}
taking values in $V_{\bf{m}}$ is a solution to the confluent KZ
equation \eqref{CKZ1}-\eqref{CKZ3}.
\end{thm}



Let $\mathfrak{S}_{m_p}$ be the group of all permutations on the
variables $\{t_a^{(p)}|a\in S_p\}$, for $p=1,\ldots, N-1$. 
We  assume that for any rational function $\varphi$ whose poles are in
the diagonal set
 $D$ defined as 
 \begin{equation}
 D=\bigcup_{a,p\atop b,q}\{t_a^{(p)}=t_b^{(q)}\}\cup
 \bigcup_{a,p,i }\{t_a^{(p)}=z_i\}\cup
 \bigcup_{i,j}\{z_i=z_j\},
 \end{equation}
 an integral formula
\begin{equation}
\int_{\Gamma}\prod\limits_{1\leq
p\leq N-1,\atop 1\leq a\leq
m_p}dt^{(p)}_a\Phi^{1/\kappa}(z,t)\varphi
\end{equation}
is invariant under the action of $\sigma=(\sigma_1,\ldots,\sigma_{N-1})$
($\sigma_p\in\mathfrak{S}_{m_p}$, $1\leq p\leq N-1$).


We also assume that for any variable $t_b^{(q)}$,
\begin{equation}
\int_{\Gamma}\prod_{1\leq
p\leq N-1,\atop 1\leq a\leq
m_p}dt^{(p)}_a
\frac{\partial}{\partial t_b^{(q)}}\left(\Phi^{1/\kappa}(z,t)\varphi\right)=0.
\end{equation}

{\bf The outline of the proof of Theorem \ref{thm solution}.} The proof follows from direct computations.
Let us explain briefly outline of the computation of the proof.
First, we compute the left hand side of the confluent KZ equation \eqref{CKZ1},
\eqref{CKZ2} and \eqref{CKZ3}.
We rewrite the result after
taking the derivations on the integral formula
by the rational functions
appeared in the expression of $f_K$ for $K\in S(m)$, using an identity
\begin{equation}
\frac{1}{(x-y)(y-z)}+\frac{1}{(y-z)(z-x)}+\frac{1}{(z-x)(x-y)}=0
\end{equation}
repeatedly and the invariance under the action of
$\mathfrak{S}_{m_1}\times\cdots\times\mathfrak{S}_{m_{N-1}}$.
Second, we compute the right hand side of the confluent KZ equation \eqref{CKZ1},
\eqref{CKZ2} and \eqref{CKZ3}.
Note that the elements $e_{-\alpha}^{(i)}$ ($\alpha\in\Delta_+$) appeared in the expression of $f_K$ for $K\in S(m)$ are ordered by
the linear order defined in Definition 2.1. We rewrite the action of the Hamiltonians  on each $f_Kv$ according to the linear order. Third, we identify the parts of the left hand side with
the parts of the right hand side. 

\medskip

We compute the case of $n=0$ and $N=3$ only in this paper. In a similar way, we can compute 
the other cases. 

Before proceeding the computation, we prepare some notations. Let 
\begin{equation*}
\nabla^{(p)}_a=\frac{\partial}{\partial t_a^{(p)}}+\frac{1}{\kappa} \frac{\partial}{\partial t_a^{(p)}}
\left(\log(\Phi)\right) \quad (p=1,2,\ a=1,\ldots,m_p). 
\end{equation*}
 Because,  for $q=1,2$ and $b=1,\ldots,m_q$, 
 \begin{equation*}
\int_{\Gamma}\prod_{1\leq
p\leq 2,\atop 1\leq a\leq
m_p}dt^{(p)}_a
\frac{\partial}{\partial t_b^{(q)}}\left(\Phi^{\frac{1}{\kappa}}\varphi\right)=
\int_{\Gamma}\prod_{1\leq
p\leq 2,\atop 1\leq a\leq
m_p}dt^{(p)}_a \Phi^{\frac{1}{\kappa}}
\nabla^{(q)}_b(\varphi).
\end{equation*}
Let rational functions $\varphi_k$ ($k=1,\ldots, \min\{ m_1,m_2\}$) be defined as
\begin{equation*}
\varphi_k=\prod_{a=1}^k \frac{1}{t_a^{(1)}-t_a^{(2)}}, 
\end{equation*}
and $\varphi_0=1$. For a rational function $\varphi(t)$,   denote by $\langle \varphi(t) \rangle$ the integral formula
\begin{equation*}
\int_{\Gamma}\prod_{1\leq
p\leq 2,\atop 1\leq a\leq
m_p}dt^{(p)}_a \Phi^{\frac{1}{\kappa}}\varphi(t). 
\end{equation*}
Then, the function $u\in V_{m}$ ($ m=(m_1,m_2)$) \eqref{eq solution} 
is represented as 
\begin{align*}
u&=\int_{\Gamma}\prod_{1\leq
p\leq 2,\atop 1\leq a\leq
m_p}dt^{(p)}_a \Phi^{\frac{1}{\kappa}}\omega_m
\\
&=(-1)^{m_1+m_2}
\sum_{k=0}^{\min\{m_1,m_2\} }
\frac{m_1!m_2!}{(m_1-k)!(m_2-k)!k!}\langle\varphi_k \rangle(e_{-\alpha_1}^{(\infty)}[1])^{m_1-k}
(e_{-\alpha_1-\alpha_2}^{(\infty)}[1])^k(e_{-\alpha_2}^{(\infty)}[1])^{m_2-k}v. 
\end{align*}

{\bf The proof of Theorem \ref{thm solution} \eqref{CKZ2} in the case of $n=0$ and $N=3$.} 
The left hand side of \eqref{CKZ2} for $p=1$ is computed as 
\begin{equation}\label{eq gamma1}
\kappa \frac{\partial u}{\partial \gamma_1}=  \sum_{b=1}^{m_1}
\langle t_b^{(1) }\omega_m\rangle. 
\end{equation}
We rewrite the result \eqref{eq gamma1} in terms of $\langle\varphi_k\rangle$. 
By using Lemma \ref{lem t},  \eqref{eq gamma1} is computed as 
\begin{align}
\sum_{b=1}^{m_1}
\langle t_b^{(1) }\omega_m\rangle
&=(-1)^{m_1+m_2}\sum_{k=0}^{\min\{m_1,m_2\} }
\frac{m_1!m_2!}{(m_1-k)!(m_2-k)!k!}
\left\{\frac{k}{\mu_1+\mu_2}\left[-(\gamma_1+\gamma_2)\langle\varphi_k\rangle
+\mu_2\langle\varphi_{k-1}\rangle\right]\right.\nonumber 
\\
&+\left.\frac{m_1-k}{\mu_1}\left[(m_2-k)\langle\varphi_{k+1}\rangle
-\gamma_1\langle\varphi_{k}\rangle\right]\right\}
(e_{-\alpha_1}^{(\infty)}[1])^{m_1-k}
(e_{-\alpha_1-\alpha_2}^{(\infty)}[1])^k(e_{-\alpha_2}^{(\infty)}[1])^{m_2-k}v. \label{eq t omega}
\end{align}
On the other hand, the right hand side of \eqref{CKZ2} for $p=1$, 
$\mathcal{H}^{(1)}_{1}u$, is easily calculated and coincides with \eqref{eq t omega}. 
For the case of $p=2$, it can be verified in the similar manner. \qed 

\begin{lem}\label{lem t}
For $0\le k \le \min\{m_1,m_2\}$ and $1\le b\le k$, we have 
\begin{equation}\label{eq nabla 1 k}
(\mu_1+\mu_2)
\langle t_b^{(1)}\varphi_k\rangle=
-(\gamma_1+\gamma_2)\langle\varphi_k\rangle
+\mu_2\langle\varphi_{k-1}\rangle, 
\end{equation}
and
\begin{equation}\label{eq nabla 2 k}
(\mu_1+\mu_2)
\langle t_b^{(2)}\varphi_k\rangle=
-(\gamma_1+\gamma_2)\langle\varphi_k\rangle
-\mu_1\langle\varphi_{k-1}\rangle.
\end{equation}
For $0\le k \le \min\{m_1,m_2\}$ and $k+1\le b\le m_1$, we have 
\begin{equation}\label{eq nabla 1 m_1}
\mu_1
\langle t_b^{(1)}\varphi_k\rangle=
(m_2-k)\langle\varphi_{k+1}\rangle
-\gamma_1\langle\varphi_{k}\rangle.
\end{equation}
For $0\le k \le \min\{m_1,m_2\}$ and $k+1\le b\le m_2$, we have 
\begin{equation}\label{eq nabla 2 m_2}
\mu_2
\langle t_b^{(2)}\varphi_k\rangle=
-(m_1-k)\langle\varphi_{k+1}\rangle
-\gamma_2\langle\varphi_{k}\rangle.
\end{equation}
\end{lem}
\begin{proof}
In order to prove \eqref{eq nabla 1 k}, we compute 
$\langle\kappa\left(\nabla_b^{(1)} + \nabla_b^{(2)}\right)\varphi_k \rangle$ as follows. 
From the definition, we have 
\begin{align*}
\left\langle\kappa\left(\nabla_b^{(1)} + \nabla_b^{(2)}\right)\varphi_k \right\rangle
=&\left\langle\left(
\sum_{c=1, c\neq b}^{m_1}\frac{2}{t_b^{(1)}-t_c^{(1)}}
+\sum_{c=1}^{m_2}\frac{-1}{t_b^{(1)}-t_c^{(2)}}
+\gamma_1+\mu_1t_b^{(1)}
\right.\right.
\\
&\left.\left.+\sum_{c=1, c\neq b}^{m_2}\frac{2}{t_b^{(2)}-t_c^{(2)}}
+\sum_{c=1}^{m_1}\frac{-1}{t_b^{(2)}-t_c^{(1)}}
+\gamma_2+\mu_2t_b^{(2)}
\right)\varphi_k\right\rangle=0. 
\end{align*}
Let $X_i$ ($i=1,2,3,4$) be defined as 
\begin{align*}
&X_1=\left\langle \sum_{c=1, c\neq b}^{m_1}\frac{2}{t_b^{(1)}-t_c^{(1)}}
\varphi_k\right\rangle,\quad 
X_2=\left\langle \sum_{c=1, c\neq b}^{m_2}\frac{-1}{t_b^{(1)}-t_c^{(2)}}
\varphi_k\right\rangle,
\\
&X_3=\left\langle \sum_{c=1, c\neq b}^{m_2}\frac{2}{t_b^{(2)}-t_c^{(2)}}
\varphi_k\right\rangle,\quad 
X_4=\left\langle \sum_{c=1, c\neq b}^{m_1}\frac{-1}{t_b^{(2)}-t_c^{(1)}}
\varphi_k\right\rangle. 
\end{align*}
By the invariance under the action of $\frak{S}_{m_1}\times \frak{S}_{m_2}$, we have 
\begin{equation*}
\frac{X_1}{2}=\left\langle \sum_{c=1,\atop c\neq b}^{k}\frac{-1}{t_b^{(1)}-t_c^{(1)}}
\varphi_k\right\rangle
+\left\langle \sum_{c=k+1}^{m_1}\frac{-1}{t_b^{(1)}-t_c^{(1)}}
\frac{1}{t_c^{(1)}-t_b^{(2)}}
\prod_{a=1, a\neq b}^{k}\frac{1}{t_a^{(1)}-t_a^{(2)}}\right\rangle. 
\end{equation*}
Hence, we obtain
\begin{equation*}
X_1+X_4
=
\left\langle \sum_{c=1, c\neq b}^{k}\frac{-1}{t_b^{(2)}-t_c^{(1)}}
\varphi_k\right\rangle.
\end{equation*}
In the similar way, we obtain 
\begin{equation*}
X_2+X_3=\left\langle \sum_{c=1, c\neq b}^{k}\frac{1}{t_b^{(1)}-t_c^{(2)}}
\varphi_k\right\rangle.
\end{equation*}
Consequently, by the invariance under the action of $\frak{S}_{m_1}\times \frak{S}_{m_2}$, 
we have 
\begin{equation*}
X_1+X_2+X_3+X_4=0. 
\end{equation*}
Therefore, we obtain 
\begin{equation*}
\langle\kappa\left(\nabla_b^{(1)} + \nabla_b^{(2)}\right)\varphi_k \rangle=
(\gamma_1+\gamma_2)\langle\varphi_k\rangle
+(\mu_1+\mu_2)
\langle t_b^{(1)}\varphi_k\rangle
-\mu_2\langle\varphi_{k-1}\rangle=0,  
\end{equation*}
which finishes the proof for the relation \eqref{eq nabla 1 k}. 

The other relations \eqref{eq nabla 2 k}, \eqref{eq nabla 1 m_1}, 
and \eqref{eq nabla 2 m_2} 
can be verified in the similar manner,  by computing 
$\langle\kappa\left(\nabla_b^{(1)} + \nabla_b^{(2)}\right)\varphi_k \rangle$, 
$\langle\kappa\nabla_b^{(1)} \varphi_k \rangle$, 
and $\langle\kappa\nabla_b^{(2)} \varphi_k \rangle$, respectively.
\end{proof}

{\bf The proof of Theorem \ref{thm solution} \eqref{CKZ3} in the case of $n=0$ and $N=3$.} 
Because, for $k=0,1,\ldots, \min\{m_1,m_2\}$,   
\begin{align*}
&\frac{\partial}{\partial \mu_1} \left(
(e_{-\alpha_1}^{(\infty)}[1])^{m_1-k}
(e_{-\alpha_1-\alpha_2}^{(\infty)}[1])^k(e_{-\alpha_2}^{(\infty)}[1])^{m_2-k}v
\right)
\\
&=\frac{1}{2}\left(
\frac{m_1-k}{\mu_1}+\frac{k}{\mu_1+\mu_2}
\right)
\left(
(e_{-\alpha_1}^{(\infty)}[1])^{m_1-k}
(e_{-\alpha_1-\alpha_2}^{(\infty)}[1])^k(e_{-\alpha_2}^{(\infty)}[1])^{m_2-k}v\right), 
\end{align*}
the left hand side of \eqref{CKZ3} for $p=1$ is computed as 
\begin{align}
\kappa \frac{\partial u}{\partial \mu_1}=&  \sum_{b=1}^{m_1}
\left\langle \frac{\left(t_b^{(1) }\right)^2}{2}\omega_m\right\rangle
+(-1)^{m_1+m_2}\sum_{k=0}^{\min\{m_1,m_2\} }
\frac{m_1!m_2!}{2(m_1-k)!(m_2-k)!k!}\nonumber 
\\
&\times
\left(
\frac{m_1-k}{\mu_1}+\frac{k}{\mu_1+\mu_2}
\right)\langle \varphi_k\rangle
(e_{-\alpha_1}^{(\infty)}[1])^{m_1-k}
(e_{-\alpha_1-\alpha_2}^{(\infty)}[1])^k(e_{-\alpha_2}^{(\infty)}[1])^{m_2-k}v. \label{eq mu1}
\end{align}
We rewrite the result \eqref{eq mu1} in terms of $\langle\varphi_k\rangle$. 
By using  Lemma \ref{lem t 2}, \eqref{eq mu1} is computed as 
\begin{align}
\kappa \frac{\partial u}{\partial \mu_1}=&
\frac{(-1)^{m_1+m_2}}{2}\sum_{k=0}^{\min\{m_1,m_2\} }
\frac{m_1!m_2!}{(m_1-k)!(m_2-k)!k!}\nonumber
\\
&\times 
\left(
\frac{k}{\mu_1+\mu_2}
\left(
\left(
-(m_1-1)+\frac{\mu_2}{\mu_1}(m_2-k+1)+\frac{(\gamma_1+\gamma_2)^2}{\mu_1+\mu_2}
\right)\langle \varphi_k \rangle
\right.\right.\nonumber
\\
&\left.-\mu_2\left(
\frac{\gamma_1}{\mu_1}+\frac{\gamma_1+\gamma_2}{\mu_1+\mu_2}
\right)\langle \varphi_{k-1}\rangle
\right)\nonumber
\\
&+\frac{m_1-k}{\mu_1}\left(
-(m_2-k)\left(
\frac{\gamma_1}{\mu_1}+\frac{\gamma_1+\gamma_2}{\mu_1+\mu_2}
\right)\langle \varphi_{k+1}\rangle
\right.\nonumber
\\
&\left.+\left(
(m_2-k)\frac{\mu_2}{\mu_1+\mu_2}+\frac{\gamma_1^2}{\mu_1}-(m_1-1)
\right)\langle\varphi_k\rangle\right)\nonumber
\\
&\times (e_{-\alpha_1}^{(\infty)}[1])^{m_1-k}
(e_{-\alpha_1-\alpha_2}^{(\infty)}[1])^k(e_{-\alpha_2}^{(\infty)}[1])^{m_2-k}v. \label{eq u_1}
\end{align}

On the other hand, the computation of the right hand side of \eqref{CKZ3} is straightforward and we see that as a result, 
the right hand side of \eqref{CKZ3} is equal to \eqref{eq u_1}. The case of $p=2$ can be proved in a similar way. 
\qed 
\begin{lem}\label{lem t 2}
For $0\le k \le \min\{m_1,m_2\}$ and $1\le b\le k$, we have
\begin{align}
(\mu_1+\mu_2)\left\langle \left(t_b^{(1)}\right)^2 \varphi_k\right \rangle 
=&\left(
-m_1+1-\kappa+\frac{\mu_2}{\mu_1}(m_2-k+1)+\frac{(\gamma_1+\gamma_2)^2}{\mu_1+\mu_2}
\right)\langle\varphi_k\rangle \nonumber 
\\
&-\mu_2\left(
\frac{\gamma_1}{\mu_1}+\frac{\gamma_1+\gamma_2}{\mu_1+\mu_2}
\right)\langle \varphi_{k-1}\rangle, \label{eq nabla t^2 1 k}
\end{align}
and 
\begin{align}
(\mu_1+\mu_2)\left\langle \left(t_b^{(2)}\right)^2 \varphi_k \right\rangle 
=&\left(
-m_2+1-\kappa+\frac{\mu_1}{\mu_2}(m_1-k+1)+\frac{(\gamma_1+\gamma_2)^2}{\mu_1+\mu_2}
\right)\langle\varphi_k\rangle\nonumber 
\\
&+\mu_1\left(
\frac{\gamma_2}{\mu_2}+\frac{\gamma_1+\gamma_2}{\mu_1+\mu_2}
\right)\langle \varphi_{k-1}\rangle. \label{eq nabla t^2 2 k}
\end{align}
For $0\le k \le \min\{m_1,m_2\}$ and $k+1\le b\le m_1$, we have 
\begin{align}
\mu_1\left\langle \left(t_b^{(1)}\right)^2 \varphi_k \right\rangle=&
-(m_2-k)\left(
\frac{\gamma_1}{\mu_1}+\frac{\gamma_1+\gamma_2}{\mu_1+\mu_2}
\right)\langle \varphi_{k+1}\rangle\nonumber
\\
&+\left(
(m_2-k)\frac{\mu_2}{\mu_1+\mu_2}+\frac{\gamma_1^2}{\mu_1}-m_1+1-\kappa
\right)\langle\varphi_k\rangle.\label{eq nabla t^2 1 m_1}
\end{align}
For $0\le k \le \min\{m_1,m_2\}$ and $k+1\le b\le m_2$, we have 
\begin{align}
\mu_2\left\langle \left(t_b^{(2)}\right)^2 \varphi_k \right\rangle=&
(m_1-k)\left(
\frac{\gamma_2}{\mu_2}+\frac{\gamma_1+\gamma_2}{\mu_1+\mu_2}
\right)\langle \varphi_{k+1}\rangle\nonumber
\\
&+\left(
(m_1-k)\frac{\mu_1}{\mu_1+\mu_2}+\frac{\gamma_2^2}{\mu_2}-m_2+1-\kappa
\right)\langle\varphi_k\rangle.\label{eq nabla t^2 2 m_2}
\end{align}
\end{lem}
\begin{proof}
In order to prove \eqref{eq nabla t^2 1 k}, we compute $\langle \kappa (\nabla_b^{(1)}
t_b^{(1)}+\nabla_b^{(2)}
t_b^{(2)})\varphi_k\rangle$ as follows. 
Let $X_i$ be defined as 
\begin{align*}
&X_1=\left\langle \sum_{c=1, c\neq b}^{m_1}\frac{2t_b^{(1)}}{t_b^{(1)}-t_c^{(1)}}
\varphi_k\right\rangle,\quad 
X_2=\left\langle \sum_{c=1, c\neq b}^{m_2}\frac{-t_b^{(1)}}{t_b^{(1)}-t_c^{(2)}}
\varphi_k\right\rangle,
\\
&X_3=\left\langle \sum_{c=1, c\neq b}^{m_2}\frac{2t_b^{(2)}}{t_b^{(2)}-t_c^{(2)}}
\varphi_k\right\rangle,\quad 
X_4=\left\langle \sum_{c=1, c\neq b}^{m_1}\frac{-t_b^{(2)}}{t_b^{(2)}-t_c^{(1)}}
\varphi_k\right\rangle. 
\end{align*}
Then, we have
\begin{align}
\left\langle \kappa \left(\nabla_b^{(1)}
t_b^{(1)}+\nabla_b^{(2)}
t_b^{(2)}
\right)\varphi_k\right\rangle=&X_1+X_2+X_3+X_4+\langle\gamma_1 t_b^{(1)}\varphi_k\rangle
+\langle\gamma_2 t_b^{(2)}\varphi_k\rangle\nonumber
\\
&+(\mu_1+\mu_2)\left\langle (t_b^{(1)})^2 \varphi_k\right\rangle
+\left\langle (t_k^{(1)}+t_k^{(2)}) \varphi_{k-1}\right\rangle+
\kappa\langle \varphi_k \rangle.\label{eq nabla for t 2}
\end{align}
We need to compute  $\sum_{i=1}^4X_i$ only,  due to Lemma \ref{lem t}. 
By the invariance under the action of $\frak{S}_{m_1}\times \frak{S}_{m_2}$, we have 
\begin{equation*}
\frac{X_1}{2}=\left\langle \sum_{c=1,\atop c\neq b}^{k}\frac{-t_c^{(1)}}{t_b^{(1)}-t_c^{(1)}}
\varphi_k\right\rangle
+\left\langle \sum_{c=k+1}^{m_1}\frac{-t_c^{(1)}}{t_b^{(1)}-t_c^{(1)}}
\frac{1}{t_c^{(1)}-t_b^{(2)}}
\prod_{a=1, a\neq b}^{k}\frac{1}{t_a^{(1)}-t_a^{(2)}}\right\rangle. 
\end{equation*}
Hence, we obtain
\begin{equation*}
X_1+X_4
=(k-1)\langle \varphi_k\rangle+
\left\langle \sum_{c=1}^{k}\frac{-t_b^{(2)}}{t_b^{(2)}-t_c^{(1)}}
\varphi_k\right\rangle.
\end{equation*}
In the similar way, we obtain 
\begin{equation*}
X_2+X_3=
(k-1)\langle \varphi_k\rangle
+\left\langle \sum_{c=1}^{k}\frac{-t_b^{(1)}}{t_b^{(1)}-t_c^{(2)}}
\varphi_k\right\rangle.
\end{equation*}
Consequently, by the invariance under the action of $\frak{S}_{m_1}\times \frak{S}_{m_2}$, 
we have 
\begin{equation}\label{eq sum x}
\sum_{i=1}^4X_i=(k-2)\langle \varphi_k\rangle. 
\end{equation}
Substituting the result \eqref{eq sum x} into \eqref{eq nabla for t 2}, we obtain the relation \eqref{eq nabla t^2 1 k}. 

The other relations \eqref{eq nabla t^2 2 k}, \eqref{eq nabla t^2 1 m_1}, 
and \eqref{eq nabla t^2 2 m_2} 
can be verified in the similar manner,  by computing 
$\langle\kappa\left(\nabla_b^{(1)}t_b^{(1)} + \nabla_b^{(2)}t_b^{(2)}\right)\varphi_k \rangle$, 
$\langle\kappa\nabla_b^{(1)}t_b^{(1)} \varphi_k \rangle$, 
and $\langle\kappa\nabla_b^{(2)}t_b^{(2)} \varphi_k \rangle$, respectively.
\end{proof}

\section{Monodromy preserving deformation}

As mentioned in the introduction, our confluent KZ equations may be
viewed as a quantization of Monodromy preserving deformation. In
this section we give the explicit correspondence following the
process in \cite{JNS}. Below we use the parameter $\hbar=1/\kappa$
in place of $\kappa$.

Recall that the confluent KZ equations \eqref{CKZ1}-\eqref{CKZ3} are
defined from the follwing date: collection of Verma modules
$V^{(i)}$ attached to each $z_i$ ($i=1,\ldots,n$), and confluent
verma module $V^{(\infty)}$ at $\infty$. Let $U$ be an invertible
matrix solution to this system. We enlarge these date by adjoining
the natural representation $\C^N$ of $\sl_N$ at the point $z=z_0$
with Poincar\'e rank $0$. Let $\widetilde{U}$ be the matrix solution
to the corresponding system \eqref{CKZ1}-\eqref{CKZ3}. Let us
consider the quantity $Y(z)=U^{-1}\widetilde{U}$. The following
equations immediately follow from the confluent KZ equations:
\begin{align}
&\frac{\partial}{\partial z}Y=A(z)Y,\label{schlesinger1}\\
&\frac{\partial}{\partial z_i}Y=B^{(i)}_{-1}Y,\label{schlesinger2}\\
&\frac{\partial}{\partial\gamma_p}Y=B^{(\infty)}_{0,p}Y,\label{schlesinger3}\\
&\frac{\partial}{\partial\mu_p}Y=B^{(\infty)}_{1,p}Y.\label{schlesinger4}
\end{align}
Here we have set
\begin{align*}
&A(z)=\hbar U^{-1}\widetilde{G}^{(0)}_{-1}U\\
&\quad\quad=\hbar
U^{-1}\left(\sum^n_{j=1}\frac{\Omega^{(j)}}{z-z_j}-\sum^{N-1}_{p=1}\gamma_pw_p-\sum_{\alpha\in\Delta}
e_{\alpha}^{(\infty)}[1]e_{-\alpha}-z\sum^{N-1}_{p=1}\mu_pw_p\right)U,\\
&B^{(i)}_{-1}=\hbar U^{-1}(\widetilde{G}^{(i)}_{-1}-G^{(i)}_{-1})U\\
&\quad=-\hbar U^{-1}\frac{\Omega^{(i)}}{z-z_i}U,\\
&B^{(\infty)}_{1,p}=\hbar U^{-1}(-\mathcal{H}_p^{(1)}+\widetilde{\mathcal{H}}_p^{(1)})U\\
&\quad=\hbar U^{-1}\left(-zw_p-\sum_{\alpha\in
J_p}\frac{1}{\mu_{\alpha}}\left(e_{-\alpha}^{(\infty)}[1]e_{\alpha}+e_{\alpha}^{(\infty)}[1]e_{-\alpha}\right)\right)U,\\
&B^{(\infty)}_{2,p}=\hbar U^{-1}(-\mathcal{H}_p^{(2)}+\widetilde{\mathcal{H}}_p^{(2)})U\\
&\quad\quad=\frac{\hbar}{2}U^{-1}\left(-z^2w_p-\sum_{\alpha\in
J_p}\frac{z}{\mu_{\alpha}}\left(e^{(\infty)}_{-\alpha}[1]e_{\alpha}+e^{(\infty)}_{\alpha}[1]e_{-\alpha}\right)+\sum_{\alpha\in
J_p}\frac{1}{\mu^2_{\alpha}}e^{(\infty)}_{-\alpha}[1]e^{(\infty)}_{\alpha}[1]h_{\alpha}\right.\\
&\left.\quad\quad+ \sum_{\alpha\in
J_p}\frac{\gamma_{\alpha}}{\mu^2_{\alpha}}\left(e^{(\infty)}_{-\alpha}[1]e_{\alpha}+e^{(\infty)}_{\alpha}[1]e_{-\alpha}\right)
+\sum_{\alpha\in
J_p}\frac{1}{\mu_{\alpha}}(e_{-\alpha}e_{\alpha}+e_{-\alpha}E_{\alpha}+E_{-\alpha}e_{\alpha})\right)U.
\end{align*}
The integrability condition for
\eqref{schlesinger1}-\eqref{schlesinger4} gives rise to a system of
non-linear differential equations with respect to the \lq time'
variables $\bf{z}$ and $\gamma$, $\mu$. These are the quantization
of the irregular Schlesinger  equations.

\begin{lem}
The quantized Schlesinger system given above are Hamiltonian
equations with the time-independent Hamiltonians
\begin{align*}
H^{(i)}_{-1}=\hbar U^{-1}G^{(i)}_{-1}U
\end{align*}
for $i=1,\ldots,n$, and
\begin{align*}
H^{(\infty)}_{1,p}=\hbar U^{-1}\mathcal{H}^{(1)}_{p}U,\\
H^{(\infty)}_{2,p}=\hbar U^{-1}\mathcal{H}^{(2)}_{p}U
\end{align*}
for $p=1,\ldots,N-1$.
\end{lem}
{\noindent\bf Proof:} The proofs follow from the integrability
conditions of \eqref{schlesinger1}-\eqref{schlesinger4} and the
integrability condition of the confluent equations
\eqref{CKZ1}-\eqref{CKZ3}.

We consider the case $\lambda=i$ ($i=1,\ldots,n$). The integrability
condition of the confluent equations \eqref{CKZ1}-\eqref{CKZ3} gives
$[\widetilde{G}^{(0)}_{-1},\widetilde{G}^{(i)}_{-1}]=0$, thus we get
$[\widetilde{G}^{(0)}_{-1},\widetilde{G}^{(i)}_{-1}-G^{(i)}_{-1}]=[\widetilde{G}^{(0)}_{-1},G^{(i)}_{-1}]$.
Conjugating with $U^{-1}$ we obtain
\begin{align*}
[H^{(i)}_{-1},A(z)]=-[B^{(i)}_{-1},A(z)].
\end{align*}

The integrability condition of \eqref{schlesinger1},
\eqref{schlesinger2} implies
\begin{align*}
0=[\frac{\partial}{\partial z}-A(z),\frac{\partial}{\partial
z_i}-B^{(i)}_{-1}]=\frac{\partial A(z)}{\partial z_i}-\frac{\partial
B^{(i)}_{-1}}{\partial z}+[A(z),B^{(i)}_{-1}]].
\end{align*}
We note that
\begin{align*}
\frac{\partial B^{(i)}_{-1}}{\partial z}=\frac{\hbar
U^{-1}\Omega^{(i)}U}{(z-z_i)^2}=\overline{\frac{\partial}{\partial
z_i}}A(z),
\end{align*}
where $\overline{\frac{\partial}{\partial z_i}}$ means
$\frac{\partial}{\partial z_i}$ acting only on the time variables,
regarding the dynamical variables as constant. In summary, we get
\begin{align*}
\frac{\partial}{\partial z_i}A(z)=\overline{\frac{\partial}{\partial
z_i}}A(z)+[A(z),H^{(i)}_{-1}],
\end{align*}
showing that $H^{(i)}_{-1}$ is the Hamiltonian for the time variable
$z_i$.

 The proof for $\gamma_p$, $\mu_p$ are similar, and we note that
 for $\mu_p$, $\overline{\frac{\partial}{\partial
\mu_p}}A(z)$ means $\hbar U^{-1}\frac{\partial
G^{(0)}_{-1}}{\partial\mu_p}U$. \qed

\section{Appendix: Singular}

We used a software Singular::Plural 3-1-0 to compute the commutativity
of the Hamiltonians.
In this appendix, we attach the program we used.

\begin{verbatim}
int n,m,k,l,s,t,p,q;
LIB  "nctools.lib";

n=3;
 // n corresponds to sl(n)

int i,j;

ring R=(0,r(1..(n-1)),m(1..(n-1)),z),(e(0..(n^2-1)),E(0..(n^2-1))),Dp;

// r(i) corresponds to the parameter \gamma_i.
// m(i) corresponds to the parameter \mu_i.
// e(i) corresponds to the root vector e^{(\infty)}_{\alpha}[1], for example,
// e(1)=e^{(\infty)}_{\alpha_1}[1],
// e(2)=e^{(\infty)}_{\alpha_1+\alpha_2}[1].
// E(i) corresponds to E_{\alpha}, for example, E(1)=E_{\alpha_1},
// E(2)=E_{\alpha_1+\alpha_2}

matrix d[2*n^2][2*n^2];

 //the matrix that defines the fundamental commutator relations

 //sum mu_i+...mu_j
poly S(1);

for(i=0;i<=(n^2-1);i++){

    S(1)=0;

    p=i / n;

    q=i%n;


    if(i%n>0){

        if(p < q){

        for(k=p+1;k<=q;k++){
            S(1)=S(1)+m(k);
        }

        d[i+1,n*q+p+1]=-S(1);


        }

    }


    for(j=i+1;j<=(n^2-1);j++){

        s=j / n;

        t=j%n;

        if(q==s){

            d[i+1+n^2,j+1+n^2]=-E(p*n+t);

        }

        if(t==p){

            d[i+1+n^2,j+1+n^2]=E(s*n+q);

        }

        if(q==s and t==p){

            d[i+1+n^2,j+1+n^2]=-E(p*n+t)+E(s*n+q);

        }


    }







}



print(d);
ncalgebra(1,d);
R;

//H(p) corresponds to Hamiltonians \mathcal{H}_p^{(2)}

for(i=1;i<=(n-1);i++){
    poly H(i);poly W(i);

    for(j=1;j<=i;j++){

        W(i)=W(i)+(n-i)*E(n*(j-1)+j-1);

    }


    for(j=i+1;j<=n;j++){

        W(i)=W(i)-(i)*E(n*(j-1)+j-1);
    }



}



poly Q(1),Q(2);
poly S(2),S(3);

for(p=1;p<=(n-1);p++){



    for(k=1;k<=n^2-1;k++){
        S(1)=0;Q(1)=0;
        i=k / n;
        j=k%n;

        if(i<p and p<=j){

            for(l=i+1;l<=j;l++){
                    S(1)=S(1)+m(l);Q(1)=Q(1)+r(l);
            }


            H(p)=H(p)-(z/S(1))*(e(n*j+i)*E(k)+e(k)*E(n*j+i));

            H(p)=H(p)+(1/S(1))*E(n*j+i)*E(k);

            H(p)=H(p)+(1/(S(1)^2))*e(n*j+i)*e(k)*(E(n*i+i)-E(n*j+j));

            H(p)=H(p)+(Q(1)/(S(1))^2)*(e(n*j+i)*E(k)+e(k)*E(n*j+i));

            H(p)=H(p)+(Q(1)^2/(S(1)^3))*e(n*j+i)*e(k);

            H(p)=H(p)-(1/S(1)^3)*e(n*j+i)^2*e(k)^2;


            for(m=j+1;m<=(n-1);m++){

                S(2)=0;Q(2)=0;
                for(l=i+1;l<=m;l++){
                    S(2)=S(2)+m(l);
                    Q(2)=Q(2)+r(l);
                }

                H(p)=H(p)-(1/(S(1)*S(2)))*(Q(1)/S(1)+Q(2)/S(2))
                    *(e(k)*e(n*j+m)*e(n*m+i)+e(n*j+i)*e(m*n+j)*e(n*i+m));

                H(p)=H(p)-(1/(S(1)*S(2)))*
                    (E(k)*e(n*j+m)*e(n*m+i)+E(n*j+i)*e(m*n+j)*e(n*i+m)
                    +e(k)*E(n*j+m)*e(n*m+i)+e(n*j+i)*E(m*n+j)*e(n*i+m)
                    +e(k)*e(n*j+m)*E(n*m+i)+e(n*j+i)*e(m*n+j)*E(n*i+m));


            }

            for(m=0;m<=i-1;m++){
                S(2)=0;Q(2)=0;
                for(l=m+1;l<=j;l++){
                    S(2)=S(2)+m(l);
                    Q(2)=Q(2)+r(l);
                }

                H(p)=H(p)+(1/(S(1)*S(2)))*(Q(1)/S(1)+Q(2)/S(2))
                    *(e(k)*e(n*m+i)*e(n*j+m)+e(n*j+i)*e(n*i+m)*e(n*m+j));


                H(p)=H(p)+(1/(S(1)*S(2)))
                    *(E(k)*e(n*m+i)*e(n*j+m)+E(n*j+i)*e(n*i+m)*e(n*m+j)
                    +e(k)*E(n*m+i)*e(n*j+m)+e(n*j+i)*E(n*i+m)*e(n*m+j)
                    +e(k)*e(n*m+i)*E(n*j+m)+e(n*j+i)*e(n*i+m)*E(n*m+j));
            }


            for(m=p;m<=n-1;m++){

                S(2)=0;

                for(l=i+1;l<=m;l++){

                    S(2)=S(2)+m(l);

                }


                if(m<j){

                    H(p)=H(p)+(1/(S(1)^2*S(2)))
                    *(e(m*n+j)*e(n*j+m)*e(n*j+i)*e(k)
                    -e(n*m+i)*e(n*i+m)*e(n*j+i)*e(k));
                }


                if(m>j){

                    H(p)=H(p)+(1/(S(1)^2*S(2)))
                    *(e(m*n+j)*e(n*j+m)*e(n*j+i)*e(k)
                    -e(n*m+i)*e(n*i+m)*e(n*j+i)*e(k));

                }
            }


            for(m=0;m<=p-1;m++){

                S(2)=0;

                for(l=m+1;l<=j;l++){

                    S(2)=S(2)+m(l);

                }
                if(m<i){

                    H(p)=H(p)+(1/(S(1)^2*S(2)))
                    *(e(n*i+m)*e(n*m+i)*e(n*j+i)*e(k)
                    -e(n*j+m)*e(n*m+j)*e(n*j+i)*e(k));
            }
                if(m>i){
                    H(p)=H(p)+(1/(S(1)^2*S(2)))
                    *(e(n*i+m)*e(n*m+i)*e(n*j+i)*e(k)
                    -e(n*j+m)*e(n*m+j)*e(n*j+i)*e(k));

                }
            }

            //the part of three distinct roots
                    for(m=j+1;m<=n-2;m++){

                for(s=m+1;s<=n-1;s++){

                    S(2)=0;S(3)=0;

                    for(l=i+1;l<=m;l++){

                        S(2)=S(2)+m(l);

                    }

                    for(l=i+1;l<=s;l++){

                        S(3)=S(3)+m(l);

                    }

                    H(p)=H(p)+(1/(S(1)*S(2)*S(3)))
                            *(
                            e(k)*e(n*j+m)*e(n*m+s)*e(n*s+i)
                            +e(n*j+i)*e(n*m+j)*e(n*s+m)*e(n*i+s)

                            +e(k)*e(n*j+s)*e(n*s+m)*e(n*m+i)
                            +e(n*j+i)*e(n*s+j)*e(n*m+s)*e(n*i+m)

                            +e(n*i+m)*e(n*m+j)*e(n*j+s)*e(n*s+i)
                            +e(n*m+i)*e(n*j+m)*e(n*s+j)*e(n*i+s)
                            );


                }
            }

            for(m=1;m<=i-1;m++){
                for(s=0;s<=m-1;s++){


                    S(2)=0;S(3)=0;

                    for(l=m+1;l<=j;l++){

                        S(2)=S(2)+m(l);

                    }

                    for(l=s+1;l<=j;l++){

                        S(3)=S(3)+m(l);

                    }

                    H(p)=H(p)+(1/(S(1)*S(2)*S(3)))
                            *(
                            e(k)*e(n*j+m)*e(n*m+s)*e(n*s+i)
                            +e(n*j+i)*e(n*m+j)*e(n*s+m)*e(n*i+s)

                            +e(k)*e(n*j+s)*e(n*s+m)*e(n*m+i)
                            +e(n*j+i)*e(n*s+j)*e(n*m+s)*e(n*i+m)

                            +e(n*i+m)*e(n*m+j)*e(n*j+s)*e(n*s+i)
                            +e(n*m+i)*e(n*j+m)*e(n*s+j)*e(n*i+s)
                            );
                        }
                }

        for(m=i+1;m<=p-1;m++){
                for(s=(j+1);s<=(n-1);s++){
                    S(2)=0;S(3)=0;
                    for(l=(m+1);l<=j;l++){
                        S(2)=S(2)+m(l);
                    }

                for(l=(m+1);l<=s;l++){
                                S(3)=S(3)+m(l);
                }
                    H(p)=H(p)-((S(1)+S(3))/(S(1)*S(2)*S(3)*(S(1)+S(3)-S(2))))
                            *(
                            e(k)*e(n*j+m)*e(n*m+s)*e(n*s+i)
                            +e(n*j+i)*e(n*m+j)*e(n*s+m)*e(n*i+s)

                            +e(k)*e(n*j+s)*e(n*s+m)*e(n*m+i)
                            +e(n*j+i)*e(n*s+j)*e(n*m+s)*e(n*i+m)

                            +e(n*i+m)*e(n*m+j)*e(n*j+s)*e(n*s+i)
                            +e(n*m+i)*e(n*j+m)*e(n*s+j)*e(n*i+s));
                    }
            }
        }
    }
}

for(p=1;p<=n-1;p++){

    H(p)=n*H(p);

    H(p)=H(p)-z^2*W(p);
}

for(p=1;p<=n-1;p++){
    for(q=p+1;q<=n-1;q++){
        bracket(H(p),H(q));
    }
}
\end{verbatim}

\bigskip
\bigskip
\noindent{\it \bf  Acknowledgement.}\quad

The authors are grateful to M. Jimbo for many helpful discussions and suggestions.

\end{document}